\documentclass[conf, a4paper, 10pt]{IEEEtran}
\usepackage[utf8]{inputenc}
\usepackage{amsmath}
\usepackage{amsfonts}
\usepackage{amssymb}
\usepackage{amsthm}
\usepackage{commath}
\usepackage{graphicx}
\usepackage{bm}
\usepackage{multirow}
\usepackage{epstopdf}
\usepackage{physics}
\usepackage{mathtools}
\usepackage{cases}

\usepackage{caption}
\allowdisplaybreaks

\newtheorem{theorem}{Theorem}

\newtheorem{definition}{Definition}
\newtheorem{lemma}[theorem]{Lemma}

\newcommand\scalemath[2]{\scalebox{#1}{\mbox{\ensuremath{\displaystyle #2}}}}

\makeatletter
\renewcommand*\env@matrix[1][\arraystretch]{%
  \edef\arraystretch{#1}%
  \hskip -\arraycolsep
  \let\@ifnextchar\new@ifnextchar
  \array{*\c@MaxMatrixCols c}}
\makeatother

\title{Lossy Transmission of Correlated Sources over Two-Way Channels}

\author{Jian-Jia~Weng, Fady~Alajaji, and Tam{\'a}s~Linder%
\thanks{The authors are with the Department of Mathematics and Statistics, Queen’s University, Kingston, ON K7L 3N6, Canada (email: jian-jia.weng@queensu.ca, fa@queensu.ca, linder@mast.queensu.ca).}
\thanks{This work was supported in part by NSERC of Canada.}}

\begin{document}

\maketitle
\thispagestyle{empty}
\begin{abstract}
Achievability and converse results for the lossy transmission of correlated sources over Shannon's two-way channels (TWCs) are presented. 
A joint source-channel coding theorem for independent sources and TWCs for which adaptation cannot enlarge the capacity region is also established. 
We further investigate the optimality of scalar coding for TWCs with discrete modulo additive noise as well as additive white Gaussian noise.
Comparing the distortion of scalar coding with the derived bounds, we observe that scalar coding achieves the minimum distortion over both families of TWCs for independent and uniformly distributed sources and independent Gaussian sources.
\end{abstract}


\section{Introduction}
Two-way channels (TWCs) were first introduced by Shannon in \cite{Shannon61}. 
Transmission over such channels makes the best use of channel resources since two users can exchange their source data on the same frequency band simultaneously. 
Also, as the channel inputs of both users can be generated by interactively adapting to the past received signals, the receivers' knowledge about the transmitted source can be refined, which may increase the rate of successful data recovery. 
Apart from the (asymptotically) lossless transmission considered by Shannon, it is also natural to investigate the lossy counterpart in which data reconstruction is allowed within a tolerable distortion. 
Such coding schemes may play an important role in improving the efficiency of data transmission over resource-limited networks. 
In this paper, we investigate the performance of lossy transmission over noisy TWCs.

In the literature, TWCs have been studied from different points of view. 
For lossless transmission, TWCs are either viewed as a part of the multiple access channels (MACs) with feedback \cite{Ong06} or related to the compound MACs with correlated side information at the receiver \cite{Gunduz09}.
Inner and outer bounds for the transmission rate over noisy TWCs were derived based on these channel models. 
Furthermore, from the channel capacity perspective, it has been found that adaptation is not always useful \cite{Han84}-\cite{Lin16}. 
In contrast to lossless transmission, studies regarding the lossy counterpart are limited.
The first lossy transmission problem over error-free TWCs appeared in \cite{Kaspi85}, in which only one user can use the channel at each time instant. 
Interactive source coding for noiseless TWCs was considered to establish a rate distortion (RD) region of TWC.
In \cite{Maor06}, these results were extended to noisy TWCs by also adopting the interactive protocol of \cite{Kaspi85}. 
To date, the performance of lossy transmission under Shannon's set-up of simultaneous user transmissions is not fully known.  

In the first part of this paper, we establish achievability and converse theorems for the lossy transmission of two correlated sources over TWCs under Shannon's scenario. 
A TWC is viewed as two one-way channels with associated states, and two-dimensional distortion regions for the TWC are derived. 
For independent sources and for TWCs whose capacity region are not enlarged by adaptation coding, we further find that the achievability and converse parts are matched, resulting in a complete joint source-channel coding theorem that shows that the TWC system can be treated as two parallel one-way systems with separate source and channel coding.  
Based on these results, we investigate the performance of scalar coding for two important classes of additive-noise TWCs: $q$-ary discrete additive-noise TWCs \cite{Lin16} and additive white Gaussian noise (AWGN) TWCs \cite{Han84}. 
For these channels, adaptation coding does not enlarge the capacity region. 
Scalar coding (also known as single-letter coding or uncoded transmission \cite{Gastpar03}) is  particularly interesting because it is the simplest one among all possible transmission schemes. 
We analyze the distortion incurred by the scalar coding scheme and compare it with the distortion lower bounds obtained from the converse theorem.
As expected, it is observed that scalar coding is sub-optimal with a performance deteriorating as the correlation between the two sources increases.
However, when the two sources are independent, we show that scalar coding is optimal for uniform sources under the Hamming distortion measure over the discrete additive TWC and for Gaussian sources under the squared error distortion measure over the AWGN-TWC.
These results are extensions of their well-known counterparts for one-way point-to-point systems (e.g., see \cite{Gastpar03} and references therein). 

The paper is organized as follows. 
In Section~II, the system model of TWCs is introduced. 
Achievability and converse results on lossy transmission over TWCs are also presented. 
In Section~III, the performance of lossy transmission over TWCs with discrete additive noise is investigated.
Analogous results are obtained for the AWGN-TWC system in Section~IV.  
Finally, concluding remarks are given in Section~V. 

\section{Lossy Transmission over Two-way Channels}
\subsection{System Model}
For $i=1, 2$, let $U_i$, $X_i$, and $Y_i$ denote random variables corresponding to terminal~$i$'s source, channel input, and channel output, respectively. 
Let $\mathcal{U}_i$, $\mathcal{X}_i$, and $\mathcal{Y}_i$ respectively denote their alphabets. 
We consider a correlated source transmission problem over Shannon's TWCs as shown in Fig.~\ref{TWC}, where two terminals want to exchange the correlated sources $U_1$ and $U_2$ within desired distortions via a memoryless and noisy TWC governed by the channel input-output transition probability $p(y_1, y_2|x_1, x_2)$. 
One special feature of such transmission is that channel inputs can be generated by adapting to the previously received signals, which may improve the quality of reconstruction. 
Let $R_i$ denote the channel coding rate of terminal~$i$ for $i=1, 2$. 
In \cite{Shannon61}, Shannon derived inner and outer bounds for the capacity region of the TWC.
Both bounds are of the same form but with different input distribution restrictions.
In particular, the bounds are of the form: 
\begin{equation}
\left\{
\begin{array}{c}
R_1\le I(X_1; Y_2|X_2),\\
R_2\le I(X_2; Y_1|X_1),
\end{array}
\right.
\label{capacity}
\end{equation}
where $I(X_i; Y_j|X_j)$ denotes conditional mutual information, and $X_1$ and $X_2$ are independent inputs in the inner bound,  while in the outer bound $X_1$ and $X_2$ are arbitrarily correlated.
Note that the inner bound is proved by a standard coding scheme which does not use adaptation. 
A channel symmetry condition for which the two regions coincide was also provided in \cite[Sec. 12]{Shannon61}. 
This result shows that adaptation coding cannot enlarge the capacity region of symmetric TWCs. 

\begin{figure}[!tb]
\centering
\includegraphics[draft=false, scale=0.44]{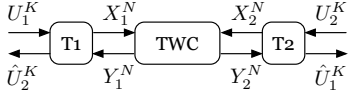} 
\caption{The block diagram of two-way communications.}
\label{TWC}
\end{figure}

For a positive integer blocklength $K$, let $U_i^{K}\triangleq(U_{i1}, U_{i2}, \dots, \allowbreak U_{iK})$ denote the source sequence of terminal~$i$, $i=1, 2$. 
Here, $U_1^K$ and $U_2^K$ are distributed according to the product probability distribution $\prod_{m=1}^K p_{\text{s}}(u_{1m}, u_{2m})$, where $(u_{1m}, u_{2m})\in\mathcal{U}_1\times\mathcal{U}_2$. 
In other words, the joint source $\{(U_{1m}, U_{2m})\}_{m=1}^K$ is memoryless in time with a joint probability distribution $p_{\text{s}}$ over $\mathcal{U}_1\times\mathcal{U}_2$ at each time instant.
For transmitting $U_{i}^K$, an encoding function $f_{in}:\mathcal{U}_{i}^K\times\mathcal{Y}_i^{n-1}\rightarrow\mathcal{X}_i$ is used by terminal~$i$ to generate the $n$th channel input for $n=1, 2, \dots, N$, where $N$ is a positive integer (note that $N=N_K$ as $N$ is a function of $K$). 
Let $X_i^{N}\triangleq(X_{i1}, X_{i2}, \dots, \allowbreak X_{iN})$ denote the channel input sequence corresponding to $U_{i}^K$, and let $Y_{i}^N\triangleq(Y_{i1}, Y_{i2}, \dots, \allowbreak Y_{iN})$ denote the received sequence at terminal~$i$. 
Specifically, we have $X_{in}=f_{in}(U_i^K, Y_i^{n-1})$ for $n=1, 2, \dots, N$, where $Y_i^{n-1}$ denotes the sequence comprising the first $n-1$ entries of $Y_i^{N}$. 
The joint probability distribution of all random vectors associated with the system is given by $p(u_1^K, u_2^K, x_1^N, x_2^N, y_1^N, y_2^N)= p_{\text{s}}(u_1^K, u_2^K)\allowbreak\prod_{n=1}^N p(x_{1n}, x_{2n}|u_1^K, u_2^K, y_1^{n-1}, y_2^{n-1})p(y_{1n}, y_{2n}|x_{1n}, x_{2n})$. 

To reconstruct the source sequence transmitted from the other terminal, terminal~$i$ uses decoder $g_{i}: \mathcal{U}_{i}^K\times\mathcal{Y}_i^N\rightarrow\mathcal{U}_j^K$ to produce the estimate $\hat{U}_j^K=g_{i}(U_i^K, Y_i^N)$ of $U_j^K$ for $i\neq j$. 
The fidelity of the reconstruction is given by $d(u_i^K, \hat{u}_i^K)\triangleq\allowbreak K^{-1}\sum_{n=1}^K d(u_{in}, \hat{u}_{in})$, where $d(u_{in}, \hat{u}_{in})$ is a single-letter distortion measure. 
The pair $(f_i, g_j)$ for $i\neq j$ constitutes a (one-way) joint source-channel code with rate $r\triangleq K/N$ (in source symbols/channel symbol), where $f_i\triangleq (f_{i1}, f_{i2}, \dots, f_{iN})$.
The associated expected distortion is given by $D_{f_i}\triangleq\mathbb{E}[d(U_i^K, \hat{U}_i^K)]=K^{-1}\sum_{n=1}^K E[d(U_{1n}, \hat{U}_{1n})]$. 
The code pairs $(f_1, g_2)$ and $(f_2, g_1)$ then form an overall coding scheme for the two-way source-channel (TWSC) system. 

\begin{definition}
A distortion pair $(D_1, D_2)$ is achievable at rate $r$ for the TWSC system if there exists a sequence of encoding and decoding functions with $\lim_{K\rightarrow\infty} K/N_{K} = r$ such that $\limsup_{K\rightarrow\infty}\mathbb{E}[d(U_i^K, \hat{U}_i^K)]\le D_i$ for $i=1, 2$. 
\end{definition}

\begin{definition}
The distortion region of a rate-$r$ TWSC system is defined as the convex closure of the set of all achievable distortion pairs. 
\end{definition}


\subsection{Necessary and Sufficient Lossy Transmission Conditions}
In \cite[Section VIII]{Gunduz09}, achievability and converse theorems for (asymptotically) lossless transmission of correlated sources in the sense that $\lim_{K\rightarrow\infty} p(U^K \neq \hat{U}^K)=0$ are derived. 
Here, we present a similar result for lossy transmission in which nonzero expected single-letter distortion is considered.
We take Shannon's viewpoint: a TWC can be viewed as two one-way channels with state variables and these states are known at the receiver but not at the transmitter \cite{Shannon61}. 
For transmitting over such a one-way channel, it is natural to treat the correlated source at the receiver as side information for the source at the transmitter as in the setting of the Wyner-Ziv coding problem \cite{WZ76}. 
Following this perspective, let $R^{(i)}(D)$ and $R^{(i)}_{\text{WZ}}(D)$ denote the standard and Wyner-Ziv RD functions of $U_i$ for $i=1, 2$, which are respectively given by
(e.g., \cite{TC91})
\[
R^{(i)}(D)=\min\limits_{p(\hat{u}_i|u_i): \mathbb{E}[d(U_i, \hat{U}_i)]\le D} I(U_i; \hat{U}_i)
\]
and
\[
R^{(i)}_{\text{WZ}}(D)=\min\limits_{p(w|u_i)}\min\limits_{{g_j: \mathcal{U}_j\times\mathcal{W}\rightarrow\mathcal{U}_i}\atop{\mathbb{E}[d(U_i, g_j(U_j, W))]\le D}} I(U_i; W|U_j)
\]
where $j\neq i$ and $W$ is an auxiliary random variable with alphabet $\mathcal{W}$ such that $|\mathcal{W}|\le |\mathcal{U}_i|+1$. 
Below, we establish inner and outer bounds on the distortion region of our lossy transmission system.

\begin{lemma}[\textbf{Achievability}]
For the rate-$r$ lossy transmission of the joint source $(U_1, U_2)$ over the memoryless TWC, the distortion pair $(D_1, D_2)$ is achievable if 
\begin{equation}
\left\{ 
\begin{IEEEeqnarraybox}[][c]{l}
\IEEEyesnumber
r\cdot R_{\text{WZ}}^{(1)}(D_1)< I(X_1; Y_2|X_2),\\
r\cdot R_{\text{WZ}}^{(2)}(D_2)< I(X_2; Y_1|X_1),\label{Ashe}
\end{IEEEeqnarraybox}
\right.
\end{equation}
for some joint probability distribution $p(u_1, u_2, x_1, x_2)=p_{\text{s}}(u_1, u_2)p(x_1)p(x_2)$ of $(U_1, U_2, X_1, X_2)$. 
\label{Ash}
\end{lemma}
\begin{proof}[Sketch of the Proof]
The idea is to exhibit a coding scheme which integrates two simultaneous separate source-channel codes, where each code in one direction of the TWSC system achieves $D_i$, $i=1, 2$. 
For given $\epsilon_i>0$, we first apply Wyner-Ziv coding with rate $R_{\text{S}, i}\triangleq K(R^{(i)}_{\text{WZ}}(D_i/(1+\epsilon_i)))$ to compress $U_{i}^K$. 
The index of the resulting codeword is then encoded by a channel code with bloklength $N_K$ for point-to-point channels without feedback.
From the capacity region inner bound in (\ref{capacity}), we know that for $i\neq j$ the index can be reliably transmitted from terminal~$i$ to $j$ if $\lim_{K \to \infty} R_{S,i}/N_K < I(X_i:Y_j|X_j)$, which implies the desired results.
\end{proof}

\begin{lemma}[\textbf{Converse}]
For the rate-$r$ lossy transmission of joint source $(U_1, U_2)$ over the memoryless TWC, if $(D_1, D_2)$ is achievable, then 
\begin{equation}
\left\{ 
\begin{IEEEeqnarraybox}[][c]{l}
\IEEEyesnumber\label{ob}
R^{(1)}(D_1)\le  I(U_1; U_2) + r^{-1}\cdot I(X_1; Y_2|X_2),\\ 
R^{(2)}(D_2)\le  I(U_1; U_2) + r^{-1}\cdot I(X_2; Y_1|X_1),
\end{IEEEeqnarraybox}
\right.
\end{equation}
for some joint probability distribution $p(u_1, u_2, x_1, x_2)=p_{\text{s}}(u_1, u_2)p(x_1, x_2)$ of $(U_1, U_2, X_1, X_2)$. 
\label{lemma3}
\end{lemma}
\begin{proof}
Suppose that there exists a source-channel coding scheme $(f_1, f_2, g_2, g_2)$ for the TWSC system with rate $r$ and average distortions $\limsup_{K \to \infty} D_{f_i} \leq D_i$, $i=1,2$. 
From terminal~$1$ to $2$, we have 
{\allowdisplaybreaks
\begin{IEEEeqnarray}{rCl}
& &K\cdot R^{(1)}(D_{f_1})\nonumber\\
&= &K\cdot R^{(1)}\left(K^{-1}\sum\limits_{n=1}^K \mathbb{E}\left[d(U_{1n}, \hat{U}_{1n})\right]\right)\nonumber\\
&\le & K\sum\limits_{n=1}^K K^{-1} R^{(1)}\left(\mathbb{E}[d(U_{1n}, \hat{U}_{1n})]\right)\label{c-3}\\
&\le &\sum\limits_{n=1}^K I(U_{1n}; \hat{U}_{1n})\label{c-5}\\
&\le &I(U_1^K; \hat{U}_1^K)\label{c-6}\\
&\le &I(U_1^K; U_2^K, Y_2^N)\label{c-8}\\
&= &I_{\rho} + I(U_1^K; Y_2^N|U_2^K)\label{c-9}\\
&= &I_{\rho} + \sum\limits_{n=1}^NI(U_1^K; Y_{2n}|U_2^K, Y_2^{n-1})\label{c-10}\\
&= &\scalemath{0.95}{I_{\rho} + \sum\limits_{n=1}^NH(Y_{2n}|U_2^K, Y_2^{n-1})-H(Y_{2n}|U_2^K, U_1^K, Y_2^{n-1})}\nonumber\IEEEeqnarraynumspace\\
&\le &I_{\rho} + \scalemath{0.95}{\sum\limits_{n=1}^NH(Y_{2n}|X_{2n})-H(Y_{2n}|U_2^K, U_1^K, Y_2^{n-1}, Y_1^{n-1})}\label{c-12}\IEEEeqnarraynumspace\\
&= & I_{\rho} + \sum\limits_{n=1}^NH(Y_{2n}|X_{2n})-H(Y_{2n}|X_{1n}, X_{2n})\label{c-13}\\
&= & I_{\rho} + \sum\limits_{n=1}^N I(X_{1n}; Y_{2n}|X_{2n})\nonumber\\
&= & I_{\rho} + N\sum\limits_{n=1}^N N^{-1}I(X_{1n}; Y_{2n}|X_{2n})\nonumber
\end{IEEEeqnarray}
where $I_{\rho} \triangleq I(U_1^K; U_2^K)=K\cdot I(U_1; U_2)$, (\ref{c-3}) is from the fact that $R^{(i)}(D_i)$ is convex, (\ref{c-5}) follows the definition of the RD function, (\ref{c-6}) is obtained by the independence of the $U_{1n}$'s and the fact that conditioning reduces entropy, (\ref{c-8}) is from the data processing inequality, (\ref{c-9}) and (\ref{c-10}) follow from the chain rule for mutual information, (\ref{c-12}) holds since $X_{2n}=f_{2n}(U_2^K, Y_2^{n-1})$ and conditioning reduces entropy, and (\ref{c-13}) is due to the Markov chain $(U_1^K, U_2^K, Y_{1}^{n-1}, Y_{2}^{n-1})-(X_{1n}, X_{2n})-(Y_{1n}, Y_{2n})$.  
In the same way, we also obtain
\begin{equation}
K\cdot R^{(2)}(D_{f_2})\le I_{\rho} + N\sum\limits_{n=1}^N N^{-1}I(X_{2n}; Y_{1n}|X_{1n})\nonumber\label{c-16}
\end{equation}
for the direction from terminal~$2$ to $1$. Since both $I(X_{1n}; Y_{2n}|X_{2n})$ and $I(X_{2n}; Y_{1n}|X_{1n})$ are known to be concave in $p(x_{1n}, x_{2n})$ \cite{Shannon61}, forming the mixture input distribution $p(x_1, x_2)\triangleq\allowbreak N^{-1}\sum_{n=1}^N p(x_{1n}, p_{2n})$ immediately results in $\sum_{n=1}^N N^{-1}I(X_{1n}; Y_{2n}|X_{2n})\le I(X_1; Y_2|X_2)$ and $\sum_{n=1}^N N^{-1}I(X_{2n}; Y_{1n}|X_{1n})\le I(X_2; Y_1|X_1)$.
Finally noting that $\limsup_{K \to \infty} D_{f_i} \leq D_i$, $i=1,2$ and that $R^{(i)}(D)$ is non-increasing and continuous, the proof is completed by taking the limit with respect to $K$.} 
\end{proof}

Although the distortion regions given by Lemmas \ref{Ash} and \ref{lemma3} do not match, we find that they coincide in some situations.

\begin{theorem}[\textbf{A Joint Source-Channel Coding Theorem}]
For the rate-$r$ transmission of independent sources $U_1$ and $U_2$ over a memoryless TWC for which adaptation does not enlarge the capacity region (i.e., for which the capacity region inner bound in (\ref{capacity}) is tight), a distortion pair $(D_1, D_2)$ is achievable if and only if  
\begin{equation}
\left\{ 
\begin{IEEEeqnarraybox}[][c]{l}
r\cdot R^{(1)}(D_1)\le I(X_1; Y_2|X_2),\nonumber\\
r\cdot R^{(2)}(D_2)\le I(X_2; Y_1|X_1),\nonumber
\end{IEEEeqnarraybox}
\right.
\end{equation}
for some $p(x_1, x_2)=p(x_1)p(x_2)$. 
\label{JSC}
\end{theorem}
\begin{proof}
If $U_1^K$ and $U_2^K$ are independent, then $I(U_1; U_2)=0$ in (\ref{ob}) and $R^{(i)}_{\text{WZ}}(D_i)=R^{(i)}(D_i)$ for $i=1, 2$.
Moreover, for TWCs where adaptation coding does not enlarge the capacity region, the optimal channel input distribution in the outer bound of (\ref{capacity}) is product form $p(x_1, x_2)=p(x_1)p(x_2)$. 
The rest of the proof follows from Lemmas~\ref{Ash} and \ref{lemma3}. 
\end{proof}

The achievability results of Lemma~\ref{Ash} and in Theorem~\ref{JSC} require the use of coding schemes with long codeword lengths. 
This requirement usually cannot be met in practice. 
Among all possible transmission schemes, the scalar coding (or uncoded) scheme, in which encoding and decoding are performed symbol-by-symbol and thus at a rate $r=1$, is the simplest one.
A scalar coding scheme which is optimal, i.e., it achieves the distortion limit of Lemma~\ref{lemma3} or Theorem~\ref{JSC}, is naturally appealing. 
In the next two sections, we investigate the performance of scalar coding over two important TWCs. 
To our knowledge, such results have not been reported.  
For the sake of brevity, we only present the results for one direction of transmission, i.e., from terminal~$1$ to $2$. 
The result for the reverse direction can be similarly derived. 
We close this section with the definition of scalar coding. 

\begin{definition}[\textbf{Scalar Coding}]
Set $N=K$. Scalar coding is a transmission scheme such that $X_{in}=h_{i}(U_{in})$ for some fixed function $h_{i}: \mathcal{U}_i\rightarrow\mathcal{X}_i$ and $\hat{U}_{in}=g_{j}(U_{jn}, Y_{jn})$, for $i\neq j$ and $n=1, 2, \cdots, N$.  
\end{definition}

\section{Lossy Transmission for TWCs with Discrete Additive Noise}
A TWC with $q$-ary modulo additive noise is defined as
\begin{equation}
\left\{ 
\begin{IEEEeqnarraybox}[][c]{l}
\IEEEyesnumber
Y_{1n}=X_{1n}\oplus X_{2n}\oplus Z_{1n}\nonumber\\
Y_{2n}=X_{1n}\oplus X_{2n}\oplus Z_{2n},\nonumber
\end{IEEEeqnarraybox}
\right.
\end{equation}
where $X_{1n}$, $X_{2n}$, $Z_{1n}$, $Z_{2n}\in\{0, 1, \dots, q-1\}$, $\oplus$ denotes the modulo-$q$ addition, and $\{Z_{1n}\}$ and $\{Z_{2n}\}$ are memoryless noise processes which are independent of each other and of the correlated sources. 
For $i=1, 2$ and $n=1, 2, \cdots, N$, we further assume that $\Pr(Z_{in}=0)=1-\epsilon_i$ and $\Pr(Z_{in}=j)=\epsilon_i/(q-1)$ for $j=1, 2, \dots, q-1$, where $0\le\epsilon_i \le (q-1)/q$.  
We begin by investigating the performance of scalar coding for an important class of binary correlated sources. 

\begin{subsubsection}{Correlated Binary Sources ($q=2$)}\label{3b}
Consider a joint binary source whose marginal probability distributions are uniform such that the individual sources are respectively modeled as the input and output of a binary symmetric channel with crossover probability $\delta\in[0, 1/2]$. 
For this joint source, the correlation coefficient is $\rho= 1-2\delta$ and the associated $R^{(i)}(D)$ under the Hamming distortion measure is given by \cite{TC91}
\begin{IEEEeqnarray}{l}
R^{(i)}(D) = 
\scalemath{0.96}{\left\{
\begin{array}{ll}
1 - H_{\text{b}}(D), & 0\le D\le 1/2,\\
0, & D>1/2,\\
\end{array}\right.}
\IEEEeqnarraynumspace\label{WZ2}
\end{IEEEeqnarray}
where $H_{\text{b}}(\cdot)$ is the binary entropy function.

From Lemma~\ref{lemma3}, it is known that any $r=1$ source-channel coding scheme achieving distortion pair $(D_1, D_2)$ must satisfy (\ref{ob}). Thus,   
\begin{IEEEeqnarray}{rCl}
R^{(1)}(D_1)&\le& I(U_1; U_2) + I(X_1; Y_2|X_2)\nonumber\\
&\le&1-H_{\text{b}}(\delta)+(H(Y_2)- H(Y_2|X_1, X_2))\label{x1}\IEEEeqnarraynumspace\\
&\le& 2-H_{\text{b}}(\delta)-H_{\text{b}}(\epsilon_2)\label{x2}
\end{IEEEeqnarray}
where (\ref{x1}) holds since $H(Y_2|X_2)\le H(Y_2)$ and (\ref{x2}) follows that $H(Y_2)\le 1$ and $H(Y_2|X_1, X_2)=H(Z_2|X_1, X_2)=H(Z_2)=H_{\text{b}}(\epsilon_2)$. 
Similarly, we have $R^{(2)}(D_2)\le 2-H_{\text{b}}(\delta)-H_{\text{b}}(\epsilon_1)$. 
Using (\ref{x2}) and (\ref{WZ2}), lower bounds for the system distortions $D_1$ and $D_2$ can be found numerically for given $\delta$ and $\epsilon_i$'s. 

Now, we consider the scalar coding scheme with $h_i(U_i)=U_i$ so that $X_{in}=U_{in}$ for $i=1, 2$ and $n=1, 2, \dots, N$.
For this encoder, it can be shown that the estimate $\hat{U}_{in}=Y_{jn}$, $i\neq j$, yields the optimum decoding performance, and the average distortions are given by $D_1=\epsilon_2$ and $D_2=\epsilon_1$. 
In Fig.~\ref{GAP2}, we plot the gap between the distortion lower bound and $\epsilon_2$ (for the direction from terminal~$1$ to $2$).
The numerical results show that scalar coding is sub-optimal. 
In particular, as the source correlation $\rho$ increases, the gap becomes larger. 
Also, when the quality of the channel deteriorates, the scalar coding scheme suffers a serious performance degradation. 
Nevertheless, when $U_1$ and $U_2$ are independent, i.e., $\rho=0$, scalar coding becomes optimal (with the gap in Fig.~\ref{GAP2} reducing to zero). 

In fact, this result for independent sources can be derived analytically. 
Since $H_{\text{b}}(\delta)=1$ when $\rho=0$, using (\ref{WZ2}) and (\ref{x2}) immediately gives $D_1\ge \epsilon_2$. 
Similarly, we have $D_2\ge \epsilon_1$. 
Clearly, the scalar coding scheme achieves the lower bounds and is hence optimal. 
We next show that this result also holds for non-binary independent sources.  
 
\begin{figure}[!tb]
\centering
\vspace{-0.2cm}
\includegraphics[draft=false, scale=0.36]{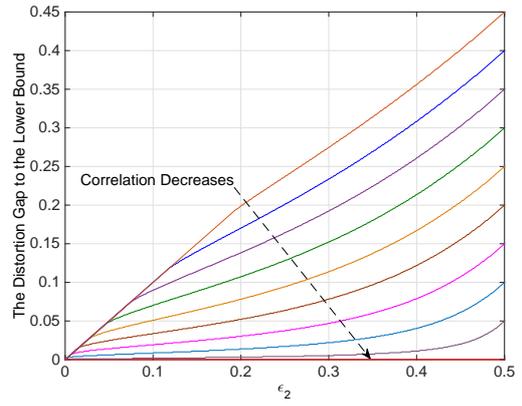} 
\caption{The performance loss of transmitting binary correlated sources via scalar coding. The curves from top to bottom correspond to $\rho$ ranging from $0.9$ to $0$ with a step size of $0.1$.}\vspace{-0.4cm}
\label{GAP2}
\end{figure}

\subsubsection{Independent and Uniformly Distributed $q$-ary Sources}
Suppose that $U_1$ and $U_2$ are independent and uniformly distributed $q$-ary sources, i.e., $\Pr(U_1=j)=\Pr(U_2=j)=1/q$ for $j=0, 1, \dots, q-1$.
In this case, $R^{(i)}(D)$ is given by \cite{Berger71}
\begin{IEEEeqnarray}{l}
\scalemath{0.85}{R^{(i)}(D)= \left\{
\begin{array}{ll}
\log_2 q-H_{\text{b}}(D)-D\log_2(q-1), & 0\le D\le \frac{q-1}{q},\\
0, & D>\frac{q-1}{q}.\\
\end{array}\right.} 
\IEEEeqnarraynumspace
\label{RC}
\end{IEEEeqnarray}
By Lemma~\ref{lemma3} (with $r=1$), one has
\begin{IEEEeqnarray}{rCl}
R^{(1)}(D_1)&\le& I(U_1; U_2) + I(X_1; Y_2|X_2)\nonumber\\
&\le&\log_2 q- H(Z_2)\nonumber\\
&=& \log_2 q - H_{\text{b}}(\epsilon_2)-\epsilon_2\log_2 (q-1)
\label{case2}
\end{IEEEeqnarray}
where the last equation is obtained by evaluating $H(Z_2)$. 
From (\ref{RC}) and (\ref{case2}), we obtain that $D_1\ge \epsilon_2$.
Similarly, we have $D_2\ge \epsilon_1$. 
On the other hand, one can easily show that the distortion achieved by the optimum decoder for scalar coding in this case is $D_1=\epsilon_2$ and $D_2=\epsilon_1$. 
Thus, scalar coding is optimal for this non-binary setting. 
\end{subsubsection}

\section{Lossy Transmission for AWGN-TWCs}
The AWGN-TWC system is described by 
\begin{equation}
\left\{ 
\begin{IEEEeqnarraybox}[][c]{l}
\IEEEyesnumber
Y_{1n} = X_{1n}+X_{2n}+Z_{1n},\nonumber\\
Y_{2n} = X_{1n}+X_{2n}+Z_{2n},\nonumber
\end{IEEEeqnarraybox}
\right.
\end{equation}
where $\{Z_{1n}\}$ and $\{Z_{2n}\}$ are memoryless zero mean Gaussian noise processes with variance $\sigma_1^2$ and $\sigma_2^2$, respectively.
Also, $\{Z_{1n}\}$ and $\{Z_{2n}\}$ are assumed to be independent of each other and of the sources. 
The $X_{in}$'s are additionally required to satisfy the power constraint $\mathbb{E}[\sum_{n=1}^{N}|X_{in}|^2]\le N\cdot P_i$, where $P_i>0$ is the average transmission power of terminal~$i$.  

The correlated sources $U_1$ and $U_2$ are herein considered to be jointly Gaussian with correlation coefficient $\rho$. 
Without loss of generality, $U_1$ and $U_2$ are assumed to have zero mean and unit variance. 
In this case, the RD function under the squared error distortion measure is given by \cite{TC91}
\begin{equation}
R^{(i)}(D) =
\begin{cases}
\frac{1}{2}\log\frac{1}{D} & 0<D\le 1,
\\
0 & D>1,
\end{cases}
\label{RDWZ}
\end{equation}
and $I(U_1; U_2)=-1/2\cdot\log (1-\rho^2)$, where $-1\le\rho\le 1$. 

We next obtain a bound on the performance limit of rate-one transmission over Gaussian TWSC systems. 
Let $\gamma_i\triangleq P_i/\sigma^2_j$ be the signal-to-noise ratio (SNR) for $i\neq j$. 
Combining (\ref{RDWZ}) with (\ref{lemma3}), we obtain the lower bounds $D_1\ge(1-\rho^2)/(1+\gamma_1)$ and $D_2\ge(1-\rho^2)/(1+\gamma_2)$.
Now, consider the scalar coding from terminal~$1$ to $2$ with $h_1$ given by $X_{1n}=h_1(U_{1n})=\alpha U_{1n}$, where $\alpha=\sqrt{P_1}$ is set to satisfy the power constraint.
At the receiver, we employ a minimum mean square error (MMSE) detector to yield the optimum estimate $\hat{U}_{1n}=\sqrt{P_1}/(P_1+\sigma^2_2)(Y_{2n}-X_{2n})$. 
From the numerical results shown in Fig.~\ref{GAP} (about the distortion gap from terminal~$1$ to~$2$), we observe a behavior similar to the discrete system of Fig.~\ref{GAP2}.
In the extreme case of $\rho=0$, i.e., when $U_1$ and $U_2$ are independent, scalar coding achieves the distortion lower bounds for both direction of transmission and is hence optimal. 

Here, for any value of $\rho$, we note that the coding scheme given in the proof of Lemma~\ref{Ash} can be used to achieve the lower bounds. 
We give a more general result in the next lemma and the achievability of the distortion lower bound for this rate-one transmission is simply obtained by setting $r=1$.

\begin{lemma}
For rate-$r$ lossy transmission of jointly Gaussian sources with zero mean, unit variance, and correlation $\rho$, over AWGN-TWCs with SNRs $\gamma_1$ and $\gamma_2$, all distortion values $D_i\ge(1-\rho^2)/(1+\gamma_i)^{1/r}$, $i=1, 2$, are achievable.   
\label{thm5}
\end{lemma}

\begin{proof}
Based on Lemma~\ref{Ash}, it suffices to check if (\ref{Ashe}) holds or not. 
First, it was shown in \cite{Han84} that the channel capacity of each direction of the two-way transmission over AWGN-TWCs is identical to the one over point-to-point AWGN channels at the same SNR. Thus, $I(X_i; Y_j|X_j)=1/2\cdot\log(1+\gamma_i)$, where $i, j=1, 2$ and $i\neq j$. 
Note that, the capacity is achieved by using independent channel inputs as required in Lemma~\ref{Ash}. 
Second, the Wyner-Ziv RD function of joint Gaussian sources under the squared error distortion measure is given by \cite{WZ76}
\begin{equation}
R^{(i)}_{\text{WZ}}(D) =
\begin{cases}
\frac{1}{2}\log\frac{1-\rho^2}{D} & 0<D\le 1-\rho^2,
\\
0 & D>1-\rho^2. 
\end{cases}
\label{RDWZ22}
\end{equation}
Let $D_i=(1-\rho^2)/((1+\gamma_i)^{1/r}-\epsilon)$ for some $\epsilon>0$. 
We immediately obtain $K/2\cdot\log ((1+\gamma_i)^{1/r}-\epsilon) < N/2\cdot\log(1+\gamma_i)$ for all $\epsilon>0$. 
Clearly, (\ref{Ashe}) holds and hence $D_i=(1-\rho^2)/(1+\gamma_i)^{1/r}$ is achievable for $i=1, 2$. 
\end{proof}

In fact, since $R^{(i)}(D)-I(U_1; U_2)=1/2 \cdot\log((1-\rho^2)/D)$ for $i=1, 2$, Lemma \ref{lemma3} can be expressed in terms of the Wyner-Ziv RD funcitons of (\ref{RDWZ22}), i.e.,   
\begin{equation}
\left\{ 
\begin{IEEEeqnarraybox}[][c]{l}
\IEEEyesnumber\label{ob2}
r\cdot R^{(1)}_{\text{WZ}}(D_1)\le  \frac{1}{2}\log(1+\gamma_1),\\ 
r\cdot R^{(2)}_{\text{WZ}}(D_2)\le  \frac{1}{2}\log(1+\gamma_2),
\end{IEEEeqnarraybox}
\right.
\end{equation}

Combining Lemma~\ref{thm5} and (\ref{ob2}), we obtain a complete joint source-channel coding theorem for Gaussian TWSC systems. 

\begin{theorem}
For the rate-$r$ lossy transmission of zero mean, unit variance, and correlation $\rho$ jointly Gaussian source $(U_1, U_2)$ over the memoryless AWGN-TWC with SNRs $\gamma_1$ and $\gamma_2$, $(D_1, D_2)$ is achievable if and only if 
\begin{equation}
\left\{ 
\begin{IEEEeqnarraybox}[][c]{l}
\IEEEnonumber
r\cdot R^{(1)}_{\text{WZ}}(D_1)\le  \frac{1}{2}\log(1+\gamma_1),\\ 
r\cdot R^{(2)}_{\text{WZ}}(D_2)\le  \frac{1}{2}\log(1+\gamma_2).
\end{IEEEeqnarraybox}
\right.
\end{equation}
\label{lemma4}
\end{theorem}

\begin{figure}[!tb]
\vspace{-0.2cm}
\centering
\includegraphics[draft=false, scale=0.36]{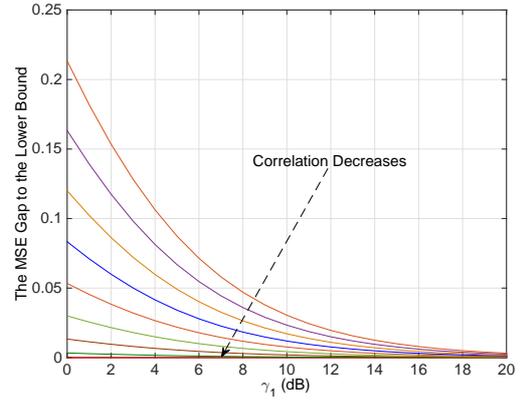} 
\caption{The performance loss of transmitting jointly Gaussian sources via scalar coding. The curves from top to bottom correspond to $\rho$ ranging from $0.9$ to $0$ with a step size of $0.1$.}\vspace{-0.4cm}
\label{GAP}
\end{figure}

\section{Conclusion}
As a first step towards understanding joint source-channel coding over Shannon's TWCs, we developed bounds on the performance of lossy transmission over TWCs. The optimality of the (simplest) scalar coding scheme is also examined for two classes of TWCs with additive noise. 
In addition to the examples given in this paper, similar results can be obtained for TWCs with erasures.  
It is observed that scalar coding is usually sub-optimal for correlated sources. 
We also provided a joint source-channel coding theorem for the lossy transmission of correlated Gaussian sources over AWGN-TWCs. 
Identifying general conditions under which two-way source-channel scalar coding is optimal is an interesting future direction.

\end{document}